\documentclass{article}
\usepackage{graphicx} 
\usepackage{amsmath}
\usepackage{amsfonts}
\usepackage{amsthm}
\usepackage[margin=1in]{geometry}
\usepackage{xcolor}
\usepackage{appendix}
\usepackage{cite}
\usepackage{mathtools}

\newcommand{\bbC}{\mathbb{C}}
\newcommand{\rmd}{\mathrm{d}}
\newcommand{\bbE}{\mathbb{E}}\newcommand{\rme}{\mathrm{e}}

\newcommand{\bbR}{\mathbb{R}}

\newcommand{\bbZ}{\mathbb{Z}}

\newcommand{\sfD}{\mathsf{D}}

\newcommand{\cN}{\mathcal{N}}

\newcommand{\cU}{\mathcal{U}}

\newcommand{\supp}{{\mathsf{supp}}}

\newtheorem{thm}{Theorem}

\newtheorem{prop}{Proposition}
\newtheorem{cor}{Corollary}
\newtheorem{lem}{Lemma}
\newtheorem{rem}{Remark}

\title{ An Improved Lower Bound on Cardinality of Support of the Amplitude-Constrained
AWGN Channel }
\author{
Haiyang~Wang%
\thanks{Haiyang Wang is with the Department of Applied and Computational Mathematics, Yale University, New Haven, CT 06511, USA (e-mail: haiyang.wang1024@gmail.com).}
\and
Luca~Barletta%
\thanks{Luca Barletta is with the Dipartimento di Elettronica, Informazione e Bioingegneria, Politecnico di Milano, 20133 Milano, Italy (e-mail: luca.barletta@polimi.it).}
\and
Alex~Dytso%
\thanks{Alex Dytso is with Qualcomm Flarion Technology, Inc., Bridgewater, NJ 08807, USA (e-mail: odytso2@gmail.com).}
}

\begin{document}

\maketitle

\begin{abstract}
We study the amplitude-constrained additive white Gaussian noise channel. 
It is well known that the capacity-achieving input distribution for this channel is discrete and supported on finitely many points. 
The best known bounds show that the support size of the capacity-achieving distribution is lower-bounded by a term of order $A$ and upper-bounded by a term of order $A^2$, where $A$ denotes the amplitude constraint. 
It was conjectured in~\cite{GaussianBoundsCard} that the linear scaling is optimal. 
In this work, we establish a new lower bound of order $A\sqrt{\log A}$, improving the known bound and ruling out the conjectured linear scaling.

To obtain this result, we quantify the fact that the capacity-achieving output distribution is close to the uniform distribution in relative entropy. 
Next, we introduce a wrapping operation that maps the problem to a compact domain and develop a theory of best approximation of the uniform distribution by finite Gaussian mixtures. 
These approximation bounds are then combined with stability properties of capacity-achieving distributions to yield the final support-size lower bound.
\end{abstract}

\section{Introduction}
We consider an additive white Gaussian noise (AWGN) channel subject to a peak-power constraint. The channel output is given by
\begin{equation}
    Y = X + Z,
\end{equation}
where the input random variable $X$ satisfies the constraint $|X| \le A$ almost surely (a.s.), and $Z$ is a standard normal random variable independent of $X$.

We are interested in the channel capacity, defined as
\begin{equation}
    C(A) = \max_{X:\,|X|\le A} I(X;Y), \label{eq:Capacity_def}
\end{equation}
where $I(X;Y)$ denotes the mutual information between $X$ and $Y$. We denote by $X^*$ a capacity-achieving input random variable and by $Y^*$ the corresponding induced output random variable. In general, both the exact value of the capacity $C(A)$ and the precise structure of the capacity-achieving distribution $X^*$ remain unknown.

This channel was first studied by Shannon in his seminal work~\cite{Shannon:1948}, where initial upper and lower bounds on the capacity were derived. Shannon also showed that, in the low signal-to-noise ratio regime, the peak-power constrained capacity exhibits the same asymptotic behavior as the capacity under an average-power (second-moment) constraint.

Subsequent progress on characterizing both the capacity and the structure of the optimal input distribution was made by Smith~\cite{smith1969Thesis,smith1971information}. In particular, Smith established that the capacity-achieving input distribution is unique, symmetric about the origin, and discrete with finitely many mass points. Moreover, Smith showed that for all $A<0.1$, the capacity-achieving distribution is equiprobable on the two-point set $\{\pm A\}$, which effectively established capacity for this regime.  This result was later sharpened by Sharma and Shamai~\cite{sharma2010transition}, who showed that an equiprobable binary input supported on $\{\pm A\}$ is optimal if and only if $A \le \bar{A} \approx 1.665$, where $\bar{A}$ is characterized as the solution to a certain integral equation. Furthermore, they demonstrated that a ternary input supported on $\{-A,0,A\}$ is capacity-achieving for all $\bar{A} \le A \le \bar{\bar{A}} \approx 2.786$. They also conjectured that, as $A$ increases, the support of the capacity-achieving distribution grows by at most one point at a time, with new mass points always appearing at zero.

Zhang~\cite{zhang1994discrete} studied the asymptotic behavior of the capacity-achieving input and output distributions. He proved bulk asymptotic uniformity of the output on $[-A,A]$~\cite[Thm.~4.16]{zhang1994discrete}, argued that the least favorable prior approaches the Jeffreys prior in the bulk with a discrete boundary pattern near $\pm A$~\cite[p.~56]{zhang1994discrete}, and, for equally weighted equally spaced priors on $[-A,A]$, showed that the optimal spacing is $\frac{2\pi(1+o(1))}{\sqrt{\ln A}}$~\cite[p.~91]{zhang1994discrete}. He also remarked that on $[0,\infty)$ the spacing of the improper asymptotic prior $W^*$ should be proportional to $1/\sqrt{\ln x}$~\cite[p.~95]{zhang1994discrete}; together with the boundary-corrected constructions on p.~96, this suggests heuristic support growth $\Theta(A\sqrt{\ln A})$, not a rigorous non-asymptotic lower bound. In contrast, our work gives explicit non-asymptotic bounds.

Dytso et al.~\cite{GaussianBoundsCard} established rigorous upper and lower bounds on the cardinality of the capacity-achieving input distribution for the amplitude-constrained AWGN channel, of orders $A^2$ and $A$, respectively. Their upper bounds rely on Karlin’s oscillation theorem~\cite{karlin1956Polya1} together with complex-analytic zero-counting arguments. The lower bound, on the other hand, relied on the fact that mutual information is bounded by the log of cardinality of input support.  Based on these results and supporting numerical evidence, the authors conjectured that the support size scales linearly with $A$.

Subsequent numerical investigations have suggested alternative asymptotic behaviors. In particular, Mattingly et al.~\cite{mattingly2018maximizing} reported an empirical scaling of order $A^{4/3}$ as $A \to \infty$, based on experiments involving not only the AWGN channel but also several non-Gaussian models, including the binomial channel and certain two-dimensional channels. A follow-up work by Abbott and Machta~\cite{abbott2019scaling} provided a heuristic, physics-inspired argument in support of this scaling. As noted earlier, this is in addition to Zhang's spacing-based heuristic of roughly $A\sqrt{\ln A}$ support growth~\cite[pp.~91,~95,~96]{zhang1994discrete}.

The diversity of these conjectured scaling laws can be partly attributed to the extreme numerical sensitivity of the underlying optimization problem. Numerical studies in this regime are highly susceptible to algorithmic bias and rounding errors, particularly for large $A$. In this limit, the nearly optimal output distribution in the interior approaches the uniform distribution, with meaningful deviations occurring only at the scale of numerical precision. Moreover, standard implementations based on the Blahut--Arimoto algorithm~\cite{blahut2003computation,arimoto1972algorithm} require nested numerical integrations of log-output probabilities, which further exacerbate numerical instability and bias. As a result, different numerical methodologies can lead to markedly different empirical scaling laws.

 In addition to studying the structure of the capacity-achieving distribution, a large body of work has focused on upper and lower bounds on the capacity in~\eqref{eq:Capacity_def}. Broadly speaking, existing capacity upper bounds fall into three categories. The first relies on the maximum entropy principle~\cite[Chapter~12]{CoverInfoTheory}, upper bounding the output differential entropy $h(Y)$ under suitable moment constraints~\cite{shamai1990capacity,dytso2019amplitude}. The second is based on a dual characterization of capacity, where maximizing mutual information over the input distribution is reformulated as minimizing relative entropy over the output distribution; a suboptimal choice of the latter yields an explicit upper bound. Notable examples in this class include the McKellips bound~\cite{mckellips2004simple} and the bound of Thangaraj and Kramer~\cite{thangaraj2017capacity}; see also~\cite{lapidoth2003capacity} for an in-depth exposition. The third approach exploits the representation of mutual information as an integral involving the minimum mean square error (MMSE)~\cite{I-MMSE}, leading to an upper bound by replacing the optimal estimator with a suboptimal one~\cite{dytso2017view}.

There also exists extensive literature on extending Smith’s original proof strategy~\cite{smith1969Thesis}, as well as the associated cardinality bounds in~\cite{GaussianBoundsCard}, to a wide range of channel models. For complex and vector Gaussian channels, discreteness of the capacity-achieving input distribution and explicit cardinality bounds are established in~\cite{ShamQuadrat,chan2005capacity,rassouli2016capacity,GaussianBoundsCard,VectorGaussianEstimationPerspect,2by2MIMO}. For additive noise channels with sufficiently regular noise densities, discreteness of the optimal input distribution is shown in~\cite{tchamkerten2004discreteness,fahs2017properties}. Proofs of discreteness and cardinality bounds for practically relevant Rayleigh fading channels can be found in~\cite{abou2001capacity,katz2004capacity,RayleighChannelCardBounds}. Similar results for non-additive channels include the Poisson channel~\cite{shamai1990capacity,PoissonCardBounds} and the binomial channel~\cite{BinomialChannel}. A broader survey of optimization-based techniques for establishing discreteness of capacity-achieving distributions is provided in~\cite{CISS2018}. Extensions to multiuser channels can be found in~\cite{mamandipoor2014capacity,ozel2015gaussian,wireatap_bounds_card}. In addition, a comprehensive overview of capacity results for point-to-point Gaussian and related channels is given in~\cite{verdu1998fifty}.

Our results also rely on recent advances in the theory of best approximation by finite Gaussian mixtures; see~\cite{BestApporximationGaussianMixture} and references therein. Indeed, our Theorem~\ref{thm:chi-sq-lower-bound} will closely parallel the arguments developed in~\cite{BestApporximationGaussianMixture}.

\subsection{Outline and Contributions}  
{
In what follows, Section~\ref{sec:main_results} presents our main results and provides some discussion. 
Section~\ref{sec:Proof_techniques} collects the main proofs. In particular, Section~\ref{sec:mixture_approximation} 
introduces a wrapping operation that maps real-valued random variables onto the circle and establishes 
lower bounds on how well a uniform distribution on $[-\pi, \pi)$ can be approximated by wrapped Gaussian 
mixtures. Section~\ref{sec:cap_properties} provides several new properties of the optimal output distribution. 
In particular, we establish an upper bound on the difference between the output distribution induced by a 
uniform input on $[-A,A]$ and that induced by $X^*$. Section~\ref{sec:main_result_proof} presents the proof 
of the main result by combining the bounds developed in the preceding sections. The proof of
Proposition~\ref{prop:cap_properties} is deferred to Appendix~\ref{sec:appendix_prop}.

To give some intuition for the proof, the argument proceeds in two steps. First, we show that if the input 
has $K$ mass points, then the induced output distribution cannot approximate the uniform distribution on 
$[-A,A]$ better than $\exp\!\big(-c (K/A)^2\big)$; this is done via a wrapping argument and bounds on 
approximation by finite Gaussian mixtures. Second, we show that the output distribution induced by the 
capacity-achieving input is within $O(1/A)$ of the uniform distribution. Putting these two statements 
together yields the desired lower bound on $K$.

Section~\ref{sec:conclusion} concludes the paper. We conclude this section by presenting relevant notation.
}

\subsection{Notation}
Throughout the paper, the deterministic scalar quantities are denoted by lower-case letters and random variables are denoted by uppercase letters.  

We denote the distribution  of a random variable $X$ by $P_{X}$. The support set of $P_X$ is denoted and defined as
\begin{equation}
\supp(P_{X})= \left\{x:  \text{ for every open set $ \mathcal{D} \ni x $ we have that $P_{X}( \mathcal{D})>0$}  \right\}. 
\end{equation} 
The notation $| \cdot |$, depending on the context, denotes either absolute value or cardinality of the set.  All logarithms are taken with base $\rme$.  The density of a standard normal will be denoted by $\phi(x) = \frac{1}{\sqrt{2 \pi}} \rme^{ -\frac{x^2}{2}}, \, x\in \bbR$. 

 Given two probability distributions $P$ and $Q$, $P\ll Q$ denotes that $P$ is absolutely continuous with respect to $Q$, and $P\ll\gg Q$ denotes that $P\ll Q$ and $Q\ll P$. Let  $p$ and $q$ be the probability density functions (pdfs) associated with $P$ and $Q$, respectively. Then, we will require the following distances: 
\begin{align}
  &\text{ Relative Entropy:} &\sfD(P \| Q) &=  \int p(x) \log \frac{p(x)}{q(x)} \rmd x,\\
  &\text{ $\chi^2$ Divergence:} &\chi^2(P \| Q) &=  \int \frac{( p(x) -q(x))^2}{q(x)} \rmd  x ,\label{eq:xi_def}
\end{align}
with the understanding that $\sfD$ and $\chi^2$ are equal to infinity if $P$ is not absolutely continuous with respect to $Q$.

\section{Main Result}
\label{sec:main_results}
The main result of this work is the following theorem. 

\begin{thm} \label{thm:main_theorem} Fix some $A>0$ and let $P_{X^*}$ be the capacity-achieving input distribution in \eqref{eq:Capacity_def}. Then,   
    \begin{equation}
        | \supp(P_{X^*})| \ge { \max \left\{ \frac{A \sqrt{\log^{+}( c\,  A)}}{2 \pi} , \, \sqrt{1 + \frac{2}{\pi \rme} A^2}, 2  \right\}\label{eq:new_bound} }
    \end{equation}
    for some explicit constant $c>0$, where $\log^{+}(x) := \max\{\log(x), 0\}$. 
\end{thm}

A few remarks are in order:
\begin{itemize}
\item  The bound in \eqref{eq:new_bound} { provides an asymptotic improvement over} the previously known lower bound of order $A$ derived in \cite{GaussianBoundsCard}. Additionally, it disproves the conjecture made in \cite{GaussianBoundsCard} that the support scales as $A$.
\item  Our result rules out linear growth of the support size and shows that any valid scaling must be superlinear. Whether the correct scaling is 
$A \sqrt{\log A}, A \log A, A^{4/3}, A^2$  or another intermediate rate remains open.
\item  {In \eqref{eq:new_bound}, the constant can be taken to be
\begin{equation}
c = \frac{1}{2 \, \zeta ( 2 \rme)}\sqrt{\frac{2}{\pi \rme}}\approx 0.0337,
\qquad
\zeta(t) = \frac{(t-1)^2}{t-1 - \log(t)}.
\end{equation}
Consequently, the logarithmic term becomes positive once $A > \frac{1}{c} \approx 29.67$. Moreover, the $\frac{A \sqrt{\log^{+}( c\,  A)}}{2 \pi}$ lower bound becomes better than the $\sqrt{1 + \frac{2}{\pi \rme} A^2}$ lower bound once $A$ is roughly $3\times 10^{5}$ or larger.}

\end{itemize}

\section{Proofs and Techniques}
\label{sec:Proof_techniques}

In this section, we collect the main proofs and techniques needed to show the main theorem.  We begin by presenting results related to the best approximation with Gaussian mixtures.  Next, we present a few new properties of the capacity-achieving input and output distributions.  The proof of Proposition~\ref{prop:cap_properties} is deferred to Appendix~\ref{sec:appendix_prop}.  This section concludes with the proof of the main theorem.

\subsection{A Best Approximation Theory With Gaussian Mixtures}
\label{sec:mixture_approximation}

We begin by introducing a wrapping operation that maps real-valued random variables onto the circle, allowing approximation questions to be studied on a compact domain. Given a continuous random variable $W \in \bbR$ and a parameter $B>0$ we define the following wrapping operator:
\begin{equation}
\left \langle W  \right\rangle_B =  \frac{\pi}{B} \left( W \text{ mod } 2B \right)  \in  [-\pi, \pi) ,
\end{equation}
where $W \text{ mod } 2B \coloneqq W -2 B \left \lfloor \frac{W+B}{2B} \right \rfloor$.

The next results summarize some properties of the wrapping operation that will be important in our derivations.  
\begin{lem} Suppose that $W \in \bbR$ is a continuous random variable with pdf $f_W$.  Then,  the following statements hold: for $B>0$
\begin{itemize}
\item the density of $\langle W \rangle_B$ is given by 
\begin{equation}\label{eq:wrapped_density_pdf}
    f _{\langle W \rangle_B } (\theta) = \frac{B}{\pi} \sum_{m \in \bbZ} f_W \left( \frac{B}{ \pi } \left(\theta + 2\pi m\right) \right),  \, \theta \in  [-\pi, \pi).
\end{equation}   
\item  Suppose that $U  \sim \mathcal{U}(-B,B)$ independent of $Z$. Then
\begin{equation}
    f _{\langle U+Z\rangle_B } (\theta) = \frac{1}{2\pi},  \, \theta \in  [-\pi, \pi). \label{eq:wrapped_uniform_pdf}
\end{equation} 
\item for any $X \in [-B,B]$ independent of $Z$, we have the following Fourier coefficients: for $ n \in \bbZ$
\begin{align}
   \widehat{f}_{\langle X+Z\rangle_B}(n)
=  \int_{-\pi}^\pi  f_{\langle X+Z\rangle_B}(\theta) \, \exp(i n \theta) \rmd  \theta=\exp\left(-\frac{1}{2} \Big(\frac{\pi n}{B}\Big)^2\right)
\bbE \left[\exp \left(  i n \langle  X \rangle_B\right) \right]. \label{eq:FoureritTransformWrapping}
\end{align}
\end{itemize}
\end{lem}
\begin{proof} We only show the last statement. 
    Fix $B>0$ and let $V=\langle X+Z\rangle_B \in [-\pi,\pi)$.
Since the map $w\mapsto \langle w\rangle_B$ is a reduction modulo $2\pi$ after scaling by $\pi/B$, and since
$\theta\mapsto \exp(in\theta)$ is $2\pi$-periodic, we have
\begin{equation}
\exp(inV)=\exp\left(in\frac{\pi}{B}(X+Z)\right)\qquad\text{a.s.}
\end{equation}
Therefore, for $n\in\bbZ$,
\begin{align}
\widehat{f}_{\langle X+Z\rangle_B}(n)
&=\int_{-\pi}^{\pi} f_{\langle X+Z\rangle_B}(\theta)\, \exp(in\theta)\,\rmd\theta \\
&=\bbE\left[\exp(inV)\right] \\
&=\bbE\left[\exp\left(in\frac{\pi}{B}(X+Z)\right)\right] \\
&=\bbE\left[\exp\left(in\frac{\pi}{B}X\right)\right]\,
\bbE\left[\exp\left(in\frac{\pi}{B}Z\right)\right],
\end{align}
where we used independence of $X$ and $Z$. Finally, since $Z\sim\cN(0,1)$,
\begin{equation*}
\bbE[\exp(itZ)] = \exp(-t^2/2),\qquad t\in\bbR,
\end{equation*}
and substituting $t=\frac{\pi n}{B}$ yields
\begin{equation}
\widehat{f}_{\langle X+Z\rangle_B}(n)
=\exp\left(-\frac{1}{2}\Big(\frac{\pi n}{B}\Big)^2\right)\,
\bbE\left[\exp\left(in\frac{\pi}{B}X\right)\right].
\end{equation}
Since $X\in[-B,B]$ implies $\langle X\rangle_B=\frac{\pi}{B}X$ a.s., the last expectation equals
$\bbE[\exp(in\langle X\rangle_B)]$, which gives \eqref{eq:FoureritTransformWrapping}.
\end{proof}

The key result for providing a lower bound is the following theorem. It quantifies the best possible approximation of the wrapped uniform distribution by a wrapped Gaussian mixture with finitely many components. Its proof adopts the trigonometric moment method of~\cite[Thm.~7]{BestApporximationGaussianMixture}; see also~\cite[Thm.~3]{wang2025super}.

\begin{thm}
\label{thm:chi-sq-lower-bound}
Let $X \in [-A,A]$ be a discrete random variable with $K>1$ mass points and $U \sim \cU[-A,A]$. Then, 
\begin{equation}
    \label{eq:lower_bound}
			\chi^2 \left(P_{\langle X+Z\rangle_A } \| P_{\langle U+Z\rangle_A } \right) \ge \frac{1}{2} \exp\left(-4\pi^2 \frac{K^2}{A^2}\right).
\end{equation}

\end{thm}

\begin{proof}
   Let us first introduce some preliminary notation. The $n$-th trigonometric moment of a random variable $W$ supported on the circle $[-\pi,\pi)$ is defined as $t_n(W) = \mathbb E [\exp(i n W)]$. Furthermore, we can define its $n$-th trigonometric moment matrix as follows:
\begin{equation}
T_n(W) = \begin{pmatrix}
t_0(W) & t_1(W) & \cdots & t_{n}(W) \\
t_{-1}(W) & t_0(W) & \cdots & t_{n-1}(W) \\
\vdots & \vdots & \ddots & \vdots \\
t_{-n}(W) & t_{-n+1}(W) & \cdots & t_{0}(W)
\end{pmatrix} ,
\end{equation}
which is an $(n+1) \times (n+1)$ Hermitian matrix. Additionally, note that, given a discrete random variable $W$ with distribution $P_W$ that has $K$ mass points, we can see that 
$T_n(W) = \sum_{k=1}^K P_W(w_k) T_n(w_k)$ is the sum of $K$ rank-$1$ matrices, therefore it has rank at most $K$.  

Now, using \eqref{eq:FoureritTransformWrapping} the Fourier coefficients of the underlying densities are given by: for $n \in \bbZ$ 
\begin{subequations}
    \begin{align}
    \widehat{f}_{\langle U+Z\rangle_A}(n)&=    \left \{\begin{array}{cc}
    1  & n=0  \\
    0 & n \neq 0
    \end{array} \right.  ,\\
\widehat{f}_{ \langle X+Z\rangle_A}(n)&   = \exp \left(-\frac{\sigma^2 n^2}{2} \right) t_n \left(  \langle X \rangle_A \right) ,
\end{align}
\label{eq:Fourier_transforms} 
\end{subequations}
where  $\sigma^2=\frac{\pi^2}{A^2}$. Consequently,  we have the following characterization of the $\chi^2$-distance
\begin{align}
    \chi^2 \left(P_{\langle X+Z\rangle_A } \| P_{\langle U+Z\rangle_A } \right) 
    &= 2\pi \int_{-\pi}^{\pi} \left(  f_{\langle X+Z\rangle_A}(\theta) - \frac{1}{2\pi} \right)^2 \rmd \theta \label{eq:using_uniformity} \\
    &= \sum_{n \in \bbZ: \, n\ne 0} \exp \left(- \sigma^2 n^2 \right) \left|t_n \left(  \langle X \rangle_A \right) \right|^2, \label{eq:parseval_application} 
\end{align}
where \eqref{eq:using_uniformity} follows from \eqref{eq:wrapped_uniform_pdf}; and \eqref{eq:parseval_application}  follows from  Parseval's theorem and Fourier expressions in \eqref{eq:Fourier_transforms}.  

Furthermore, we have that the $\chi^2$ divergence can be lower bounded as follows: 
\begin{align}
  \chi^2 \left(P_{\langle X+Z\rangle_A } \| P_{\langle U+Z\rangle_A } \right)   &\ge \sum_{n\ne0, |n| \le 2K} \exp\left(- \sigma^2 n^2 \right) \left|t_n \left(  \langle X \rangle_A \right) \right|^2   \\
    &\ge \exp(-\sigma^2 4K^2) \sum_{n\ne0, |n| \le 2K}  \left|t_n \left(  \langle X \rangle_A \right) \right|^2  \\
    &\ge \frac{\exp(-\sigma^2 4 K^2)}{2K+1} \sum_{n\ne0, |n| \le 2K} (2K+1-|n|) \left|t_n \left(  \langle X \rangle_A \right) \right|^2  \\ 
    &= \frac{\exp(-\sigma^2 4K^2)}{2K+1} \|T_{2K}(\langle X \rangle_A) - I\|_{\text{F}}^2, \label{eq:Frobinious_lower_bound}
\end{align}
where $\| \cdot \|_{\text{F}}$ is the Frobenius norm of a matrix. 

Now, as was argued already before, $T_{2K}(\langle X \rangle_A)$ is a matrix with rank $K$ at most, therefore, by Eckart-Young-Mirsky theorem~\cite{Mirsky1960}, we have that 
\begin{equation}
    \|T_{2K}(\langle X \rangle_A) - I\|_{\text{F}}^2  \ge  \min_{B \in \bbC^{(2K+1) \times (2K+1)}:\,\text{rank}(B) \le K} \| B - I\|_{\text{F}}^2 = K+1  . \label{eq:EYM_application}
\end{equation} 
Substituting \eqref{eq:EYM_application} into \eqref{eq:Frobinious_lower_bound}, we obtain:
\begin{align*}
   \chi^2 \left(P_{\langle X+Z\rangle_A } \| P_{\langle U+Z\rangle_A } \right)  \ge \frac{\exp(-\sigma^2 4 K^2)(K+1)}{2K+1}   \ge \frac{1}{2}\exp \left(-4 \pi^2 \frac{K^2}{A^2} \right),
\end{align*}
which concludes the proof of the theorem. 
\end{proof}

If, in addition, the wrapped density is bounded, then we can upgrade the above result to a lower bound on the reverse relative entropy.

\begin{cor} \label{cor:bound_on_reverse_KL}
 Suppose that the assumption of Theorem~\ref{thm:chi-sq-lower-bound} holds and assume that  $\sup_x f _{\langle X+Z\rangle_A }(x) \le M$ for some constant $M>0$.  Then, 
\begin{equation}
    \sfD\left(P _{\langle U+Z\rangle_A } \| P_{\langle X+Z\rangle_A } \right) \ge  \frac{1}{2 \, \zeta ( 2 \pi M)} \exp\left(-4\pi^2 \frac{K^2}{A^2}\right), \label{eq:lower_bound_on_KL_of_warped}
\end{equation}
where
 \begin{equation}
    \zeta(t) = \frac{(t-1)^2}{t-1 - \log(t)}.  \label{eq:zeta_definition}
\end{equation}
\end{cor}
\begin{proof}
{
The proof will require the following two inequalities: Suppose that $P \ll\gg Q$, $P \neq Q$ and $\beta_1 \in (0,1)$ where
 \begin{equation}
        \beta_1 = \inf_x  \frac{\rmd Q}{\rmd P}(x) ; \label{eq:beta_1_definition} 
    \end{equation}
    then the following inequalities hold:
\begin{enumerate}
    \item\emph{Bounding $\chi^2$ with Relative Entropy} \cite[Eq.~(169)]{sason2016f}: 
    \begin{equation}
       \frac{\chi^2(P\|Q) }{\sfD(P \|Q )} \le \frac{1}{\kappa_2(\beta_1^{-1})}   ,\label{eq:Chi_to_KL}
    \end{equation}
    where 
    \begin{equation}
        \kappa_2(t)= \frac{t \log t + (1-t)}{(1-t)^2}.
    \end{equation}
    \item \emph{Ratio of Relative Entropies} \cite[Thm.~6]{sason2016f}: 
    \begin{equation}
        \frac{\sfD(P\|Q) }{\sfD(Q \|P )} \le \kappa(\beta_1^{-1}) ,\label{eq:KL_reversing}
    \end{equation}
    with $\beta_1$ as in \eqref{eq:beta_1_definition} and 
    \begin{equation}
        \kappa(t) =  \frac{t \log t + (1-t)}{ t-1 - \log(t)}. 
    \end{equation}
\end{enumerate}
Combining inequality \eqref{eq:Chi_to_KL} and \eqref{eq:KL_reversing}, we arrive at 
\begin{equation}
    \frac{\chi^2(P\|Q)}{\sfD(Q\|P)} 
    =  \frac{\chi^2(P\|Q)}{\sfD(P\|Q)} \frac{\sfD(P\|Q)}{\sfD(Q\|P)}
    \le  \frac{\kappa(\beta_1^{-1})}{\kappa_2(\beta_1^{-1})}  = \zeta(\beta_1^{-1}) , \label{eq:combined_inequliaty}
\end{equation}
where $\zeta(t)$ is defined in \eqref{eq:zeta_definition}.

Now we apply the above inequality to our specific case and let
$P = P_{\langle X+Z\rangle_A}, Q = P_{\langle U+Z\rangle_A}$. First, note that by \eqref{eq:wrapped_density_pdf} and positivity of the Gaussian density, the density of $P$ is positive and continuous on the circle, so $P\ll\gg Q$. Second, since $Q$ has density $q(x)=\frac{1}{2\pi}$ on $[-\pi,\pi)$ and $P$ has density $p(x)=f_{\langle X+Z\rangle_A}(x)$, the bound on $p$ implies
\begin{equation}
  \beta_1=  \inf_{x \in [-\pi,\pi)}  \frac{\rmd Q}{\rmd P}(x)  \ge  \frac{1}{2 \pi M}.
\end{equation}
Since $t\mapsto\zeta(t)$ increases for $t\ge 1$, we have $\zeta(\beta_1^{-1}) \le \zeta(2\pi M)$.
Combining this with the inequality in \eqref{eq:combined_inequliaty} and the bound in Theorem~\ref{thm:chi-sq-lower-bound}, we arrive at the bound in \eqref{eq:lower_bound_on_KL_of_warped}. 
}

\end{proof}

 \subsection{Some Properties of Capacity-Achieving Distribution}
 \label{sec:cap_properties}

In this section, we summarize some of the known properties of the capacity and capacity-achieving distributions. We begin by presenting the following well-known stability result \cite{topsoe1967information,csiszar2011information}. 
 \begin{lem}\label{lem:golden} Given a channel $P_{Y|X}$, suppose that $P_{X^*}$ is a capacity-achieving distribution. Then, for any $P_X$ we have that 
    \begin{equation}
         \sfD( P_{Y} \| P_{Y^*})  \le   I(X^*; Y^* ) - I(X;Y) ,\label{eq:Golden_bound}
    \end{equation}
    where $P_X$ and $P_{X^*}$ induce $P_Y$ and $P_{Y^*}$, respectively, through the channel $P_{Y|X}$.\footnote{If the capacity-achieving distribution does not exist, $I(X^*;Y^*)$ should be replaced by the capacity value. Note that, as shown in \cite{kemperman1974shannon}, the capacity-achieving output distribution $P_{Y^*}$ always exists and is unique. In our setting, the capacity-achieving input distribution $P_X^*$ exists and is unique, as shown in \cite{smith1969Thesis}.} 
\end{lem}

A related result to Lemma~\ref{lem:golden} are the following KKT conditions \cite{smith1971information,CISSdytso2018}.
\begin{lem}\label{lem:KKT} Consider the amplitude constrained scalar additive Gaussian channel $Y = X + Z$ where the input $X$, satisfying $|X| \le A$ a.s., is independent of the noise $Z \sim \mathcal{N}(0, 1)$. The capacity-achieving distribution $P_{X^*}$ and induced output distribution $P_{Y^*}$ satisfy the following:  for $A>0$
\begin{align}
    \sfD( P_{Y|X}(\cdot| x) \| P_{Y^*})  &\le C(A), \quad x \in [-A,A],  \\
    \sfD( P_{Y|X}(\cdot| x) \| P_{Y^*})  &= C(A), \quad x\in {\rm supp}(P_{X^*}). \label{eq:equality_eq_in_KKT}
\end{align}

\end{lem}

The next result provides upper and lower bounds on the capacity. It also quantifies stability of the optimal output distribution, by bounding the relative entropy between optimal output distribution $P_{X^*+Z}$ and output distribution induced by a uniform input $P_{U+Z}$. 

\begin{lem}   Let $U \sim \cU_{[-A,A]}$. Then, for $A>0$
\begin{equation}
   \frac{1}{2} \log \left( 1 + \frac{2 A^2}{\pi \rme}\right) 
    \le I(U; U+Z) \le  C(A) \le  \log \left(1+ \frac{ \sqrt{2} A}{\sqrt{\pi \rme}} \right) .  \label{eq:Cap_bounds}
\end{equation}
Moreover, 
\begin{equation}
    \sfD ( P_{U +Z} \|  P_{X^* +Z} ) \le   \sqrt{\frac{ \pi \rme}{2} }  \frac{1}{A}. \label{eq:KL_betwen_unif_mixture_and_cap_mixture}
\end{equation}
    
\end{lem} 
\begin{proof}
The lower bound in \eqref{eq:Cap_bounds} is due to Shannon \cite[Section~25]{Shannon:1948} and the upper bound is due to McKellips \cite{mckellips2004simple}; see also \cite[Sec.~IV.A]{thangaraj2017capacity} for the complete proof of McKellips bound. To show the bound on the relative entropy of the output distributions note that
\begin{align}
\sfD ( P_{U +Z} \|  P_{X^* +Z} ) &\le
   C(A) - I(U; U+Z) \label{eq:use_golden} \\
   &\le  \log \left(1+ \frac{ \sqrt{2} A}{\sqrt{\pi \rme}} \right) - \frac{1}{2} \log \left( 1 + \frac{2 A^2}{\pi \rme}\right) \label{eq:EPI_lower_bound}\\
   &= \log  \frac{1+ \frac{ \sqrt{2} A}{\sqrt{\pi \rme}}}{ \sqrt{1 +\frac{2 A^2}{\pi \rme} }} \\
   & \le \log  \left(1 + \frac{1}{  \sqrt{\frac{2 }{\pi \rme}} A} \right)\\
   &\le  \frac{1}{  \sqrt{\frac{2 }{\pi \rme}} A} , \label{eq:bound_On_ratio}
   \end{align}
   where \eqref{eq:use_golden} follows from Lemma \ref{lem:golden};  \eqref{eq:EPI_lower_bound} follows from the bounds in \eqref{eq:Cap_bounds}; and \eqref{eq:bound_On_ratio} follows from  using inequality $1+u \le \rme^u, \, u \ge 0$.   
\end{proof}

We now produce a few new properties of the optimal distributions. 

\begin{prop}\label{prop:cap_properties} The capacity-achieving input and output distributions satisfy the following for $A>0$:
\begin{itemize}
\item[(P$_1$)] Suppose that $y_0$ is a global maximum of $f_{X^*+Z}$, then there exists a point $x \in \supp(P_{X^*})$ such that 
\begin{equation}
    | y_0-x| \le 1. \label{eq:proximity_to_max_of_f_y}
\end{equation}
    \item[(P$_2$)] Let $M_A =  \max_{ y \in \bbR}f_{X^*+Z}(y) $, then 
    \begin{equation}
   \frac{1}{\sqrt{2 \pi \rme} +2 A}\le  \rme^{ - C(A) -h(Z) }   \le  M_A   \le \rme^{ - C(A) -h(Z) +1} \le  \frac{\rme}{ \sqrt{2 \pi \rme + 4 A^2}} , \label{eq:bound_on_dencity_Y}
        \end{equation}
        where $h(Z)$ is differential entropy of $Z$.
    \item[(P$_3$)] For $|y|  \ge A$  
    \begin{equation}
        f_{X^*+Z}(y)  \le  M_A \rme^{- \frac{(|y|-A)^2}{2}} .\label{eq:Tail_bound_beyond_y>A}
    \end{equation}
   \item[(P$_4$)] 
    \begin{equation}
       {\sup_{\theta \in [-\pi, \pi) } f_{ \langle X^*+Z \rangle_A}(\theta)  \le  \frac{\rme }{ \pi} .}\label{eq:bound_on_folder_density}
    \end{equation}
\end{itemize}
    
\end{prop}
\begin{proof}
See Appendix~\ref{sec:appendix_prop}.
\end{proof}

\begin{rem}
In~\cite[Thm.~4.16]{zhang1994discrete}, Zhang showed that for any $0 < B < A$ with $A-B \to \infty$,
\begin{equation}
    \lim_{A \to \infty} \sup_{|y| \le B} \left| 2A f_{X^*+Z}(y) - 1 \right| = 0,
\end{equation}
which implies that the output distribution $P_{X^*+Z}$ is asymptotically uniform in the bulk interior of $[-A,A]$.

He also argued that the least favorable prior approaches the Jeffreys prior in the bulk with a discrete boundary correction near $\pm A$~\cite[p.~56]{zhang1994discrete}. In the equally weighted equally spaced class on $[-A,A]$, his Theorem~5.3 gives spacing $\frac{2\pi(1+o(1))}{\sqrt{\ln A}}$~\cite[p.~91]{zhang1994discrete}; on least favorable improper prior over $[0,\infty)$ he further remarks that the spacing of $W^*$ should be proportional to $1/\sqrt{\ln x}$~\cite[p.~95]{zhang1994discrete}. Together with the boundary-corrected constructions on p.~96, these heuristics point to support growth $\Theta(A\sqrt{\ln A})$~\cite[p.~96]{zhang1994discrete}.

In contrast, the bound in~\eqref{eq:bound_on_dencity_Y}, which shows that $\max_y f_{X^*+Z}(y)=\Theta(1/A)$, that is, bounded above and below by positive constant multiples of $1/A$, is stronger in the sense that it provides a uniform, non-asymptotic guarantee. On the other hand, it is weaker in that it does not recover the optimal asymptotic constant.
\end{rem}

\subsection{Proof of the Main Theorem}
\label{sec:main_result_proof}

{ We only show the first lower bound. The lower bound $\sqrt{1 + \frac{2}{\pi \rme} A^2}$ was shown in \cite{GaussianBoundsCard}. The bound $|\supp(P_{X^*})| \ge 2$ is immediate because $C(A)>0$ for every $A>0$ by \eqref{eq:Cap_bounds}, whereas any one-point input induces zero mutual information.
} 

{Denote $K =| \supp(P_{X^*})|$. From property (P$_4$) of Proposition~\ref{prop:cap_properties}, we have that $ f_{\langle X^*+Z\rangle_A} \le \frac{\rme}{ \pi}$, so $2\pi M=2\rme$ in Corollary~\ref{cor:bound_on_reverse_KL}. Therefore,
\begin{align}
\frac{1}{2 \, \zeta ( 2 \rme)} \exp\left(-4\pi^2 \frac{K^2}{A^2}\right) &\le   \sfD\left(P _{\langle U+Z\rangle_A } \| P_{\langle X^*+Z\rangle_A } \right)    \\
&\le \sfD \left( P_{ U+Z } \| P _{X^*+Z } \right) \label{eq:data_processing_inequality_reverse_kl}\\
&  \le     \sqrt{\frac{ \pi \rme}{2} }   \frac{1}{A}  , \label{eq:using_bound_on_KL_uniform_proximity}
\end{align}
where \eqref{eq:data_processing_inequality_reverse_kl} follows by data processing inequality \cite{polyanskiyInformationTheoryCoding2025}; and \eqref{eq:using_bound_on_KL_uniform_proximity} follows from using \eqref{eq:KL_betwen_unif_mixture_and_cap_mixture}. 
Rearranging the above bounds, we have that
\begin{align*}
4\pi^2 \frac{K^2}{A^2}
\ge  \log \left(  \frac{1}{2 \, \zeta ( 2 \rme)} \sqrt{\frac{2}{\pi \rme}}\, A \right).
\end{align*}
Since the left-hand side is nonnegative, we may take the positive part and conclude that
\begin{equation*}
K \ge \frac{A}{2\pi}\sqrt{\log^{+}\left(\frac{1}{2 \, \zeta ( 2 \rme)}\sqrt{\frac{2}{\pi \rme}}\, A\right)},
\end{equation*}
which proves Theorem~\ref{thm:main_theorem} with $c=\frac{1}{2 \, \zeta ( 2 \rme)}\sqrt{\frac{2}{\pi \rme}}\approx 0.0337$.

}

\section{Conclusion}
\label{sec:conclusion}

In this work, we derived a new lower bound on the cardinality of the capacity-achieving input distribution for the amplitude-constrained AWGN channel. Our result improves the previously known linear lower bound and establishes that the support size must grow superlinearly with the amplitude constraint, thereby ruling out linear scaling.

Several questions remain open. While our result shows superlinear growth, the exact asymptotic scaling, whether 
$A\sqrt{\log A}$, $A\log A$, $A^{4/3}$, $A^2$, or something in between, remains unresolved. Further refinements of the constants and a deeper understanding of the boundary behavior of the capacity-achieving distribution are of interest. Extending these techniques to other channel models, including vector Gaussian and fading channels, is another promising direction.

{ Beyond the refined bound itself, the main conceptual contribution of this work lies in the methodology
used in the proof. Our approach combines a comparison between the output distribution induced by the
capacity-achieving input and that induced by a uniform input with a wrapping argument that maps the
problem to a compact domain. This transformation allows the problem to be studied through approximation
of the uniform distribution by finite Gaussian mixtures, for which sharp lower bounds are available.
We believe that this technique may be useful for studying other structural properties of optimal input
distributions and for related problems involving Gaussian mixtures and approximation on compact
domains. In particular, it would be interesting to investigate whether a similar approach could be used
to derive bounds on the support of the least-favorable distribution in the classical problem of
estimating a bounded normal mean \cite{LeastFavorableDytso,berry1990minimax,casella1981estimating}. }

\appendix
\section{Proof of Proposition~\ref{prop:cap_properties}}
\label{sec:appendix_prop}

\begin{proof}
We now prove Proposition~\ref{prop:cap_properties}.
    We begin with (P$_1$). We claim that there exists a support point $x_k \in \supp(P_{X^*})$ such that $|y_0 - x_k| \le 1$.
			Recall $f_{X^*+Z}(y) = \sum_i w_i \phi(y-x_i)$ where $w_i > 0$. 
			Suppose $|y_0 - x_i| > 1$ for all $i$. Then $f_{X^*+Z}''(y_0) = \sum w_i ((x_i-y_0)^2 - 1) \phi(y_0-x_i) > 0$, which contradicts the fact that $y_0$ is a maximum. Thus, there must be at least one support point $x_k$ with $|y_0 - x_k| \le 1$.

        We now prove (P$_2$). Let $y_0$ be a global maximum of $f_{X+Z}$. Using a generalization of Tweedie's formula, also known as Hatsell and Nolte \cite{hatsell_nolte}, we have that \cite[Eq.~(54)]{dytso2022conditional}:
            \begin{equation}
            \frac{\rmd^2}{\rmd y^2} \log  f_{X+Z}(y)  = \text{Var}(X|X+Z=y) - 1 \ge -1,
            \end{equation}
            since $\text{Var}(X|X+Z=y) \ge 0$. Therefore,  since $y_0$ is a global maximum, for any $y$ we have that 
            \begin{equation}
                \log f_{X+Z}(y) \ge \log f_{X+Z}(y_0) - \frac{1}{2}(y-y_0)^2.  \label{eq:lower_bound_log_f}
            \end{equation}
            Therefore,  $f_{X+Z}(y) \ge M_A e^{-\frac{(y-y_0)^2}{2}}$ where $M_A:=f_{X+Z}(y_0)  $. 
            
            We now upper bound $M_A$. Starting with the KKT condition in \eqref{eq:equality_eq_in_KKT},
 we have that  for $x \in \supp(P_{X^*})$           \begin{align}
 -C(A) -h(Z) &= - \sfD( P_{Y|X}(\cdot| x) \| P_{X^*+Z})  -h(Z)\\
 &= \int_{-\infty}^{\infty} \phi(y-x) \log f_{X^*+Z}(y) \rmd y \\
 & \ge \int_{-\infty}^{\infty} \phi(y-x) \left( \log M_A - \frac{1}{2}(y-y_0)^2 \right) \rmd y \label{eq:Using_bound_concavity_bound} \\
&=  \log M_A - \frac{1}{2} \int_{-\infty}^{\infty} \phi(y-x) (y-y_0)^2 \rmd y \\
&= \log M_A - \frac{1}{2} \left( 1 + (x - y_0)^2 \right)  \\
&\ge \log M_A - 1 \label{eq:bounding_|y-x|},
 \end{align}
 where \eqref{eq:Using_bound_concavity_bound} follows from using the lower bound in \eqref{eq:lower_bound_log_f}; and \eqref{eq:bounding_|y-x|} follows from the bound in \eqref{eq:proximity_to_max_of_f_y}.  Consequently, we have that 
\begin{align}
    \log (M_A)  &\le - C(A) -h(Z) +1 \\
    & \le - \frac{1}{2} \log \left( 1 + \frac{2 A^2}{\pi \rme}\right) -h(Z) +1 \label{eq:Lower_bound_capa}\\
    &= \log \left(  \frac{\rme}{ \sqrt{2 \pi \rme + 4 A^2}} \right) ,
\end{align}
where in \eqref{eq:Lower_bound_capa} we have used the lower  bound on $C(A)$ in \eqref{eq:Cap_bounds}. This concludes the proof of the upper bounds in (P$_2$). 

To show the lower bound, we follow similar steps and note 
for $x \in \supp(P_{X^*})$       
\begin{align}
 -C(A) -h(Z) &= - \sfD( P_{Y|X}(\cdot| x) \| P_{X^*+Z})  -h(Z)\\
 &= \int_{-\infty}^{\infty} \phi(y-x) \log f_{X^*+Z}(y) \rmd y \\
 & \le   \log f_{X^*+Z}(y_0) = \log(M_A).
 \end{align}
 The proof of the lower bounds in (P$_2$) is concluded by using the upper bound $C(A)  \le \log \left(1+ \frac{ \sqrt{2} A}{\sqrt{\pi \rme}} \right)$ in \eqref{eq:Cap_bounds}.

We now prove (P$_3$). First, we establish a pointwise bound for the density in the tail region $y > A$. For any $x \in [-A, A]$ and $y > A$, we have $y-x \ge y-A > 0$. Let $y = A + t$ with $t > 0$ and note that 
\begin{align}
    (y-x)^2 - (A-x)^2 &= (A+t-x)^2 - (A-x)^2 \\
    &= (A-x)^2 + 2t(A-x) + t^2 - (A-x)^2 \\
    &= 2t(A-x) + t^2.
\end{align}
Since $x \le A$ and $t > 0$, we have $2t(A-x) \ge 0$. Therefore, $(y-x)^2 - (A-x)^2 \ge t^2 = (y-A)^2$, which implies that 
\begin{equation}
    \frac{\phi(y-x)}{\phi(A-x)} \le \exp\left( -\frac{1}{2} (y-A)^2 \right). \label{eq:upper_bound_gauss_Ration}
\end{equation}
With the upper bound \eqref{eq:upper_bound_gauss_Ration}, we have that 
\begin{align}
    f_{X^*+Z}(y) &= \int_{-A}^{A} \phi(y-x) \rmd P_{X^*} (x) \\
    &\le \int_{-A}^{A} \rme^{- \frac{(y-A)^2}{2}} \phi(A-x) \rmd P_{X^*} (x) \\
    &= \rme^{- \frac{(y-A)^2}{2}} f_{X^*+Z}(A) \\
    & \le   M_A  \rme^{- \frac{(y-A)^2}{2}},  \end{align}
    where in the last inequality we have that $M_A =  \max_{ y \in \bbR}f_{X^*+Z}(y) $.
Mirroring the same argument for $y<-A$, we arrive at the desired bound in (P$_3$). 

We now prove (P$_4$). Starting with the expression of the wrapped density in \eqref{eq:wrapped_density_pdf} 
\begin{align}
    f_{\langle X^*+Z \rangle_A}(\theta) &= \frac{A}{\pi} \sum_{k \in \mathbb{Z}} f_{X^*+Z}\left( \frac{A}{\pi}(\theta + 2\pi k) \right) \\
    &= \frac{A}{\pi}f_{X^*+Z}\left( \frac{A}{\pi}\theta  \right)  + \frac{A}{\pi} \sum_{k \in \mathbb{Z}, k \neq 0} f_{X^*+Z}\left( \frac{A}{\pi}(\theta + 2\pi k) \right)\\
    & \le  {\frac{A}{\pi} \frac{\rme}{ \sqrt{2 \pi \rme + 4 A^2}}\left(1 + \sum_{k \in \mathbb{Z}, k \neq 0} \rme^{- \frac{A^2 \left(  \left| \frac{\theta}{\pi}+2k \right|-1 \right)^2}{2}}\right)} \label{eq:boundsing_fold_with_tail_bounds}\\
    & =  {\frac{A}{\pi} \frac{\rme}{ \sqrt{2 \pi \rme + 4 A^2}}\left(1 + \sum_{m \in \mathbb{Z}} \rme^{- \frac{A^2 \left( 2m+1-\frac{\theta}{\pi} \right)^2}{2}}\right)} \label{eq:triangle_inequality_bound_exponent} \\
    & \le  {\frac{A}{\pi} \frac{\rme}{ \sqrt{2 \pi \rme + 4 A^2}}\left(2 + 2\sum_{m=1}^{\infty} \rme^{-2A^2m^2}\right)} \label{eq:bounding_gaussian_sum}\\
    & = {\frac{A}{\pi} \frac{2\rme}{ \sqrt{2 \pi \rme + 4 A^2}}\left(1+\sum_{m=1}^{\infty} \rme^{-2A^2m^2}\right)} \\
    & \le {\frac{\rme}{\pi}},
\end{align}
     {Here \eqref{eq:boundsing_fold_with_tail_bounds} follows from \eqref{eq:bound_on_dencity_Y} and \eqref{eq:Tail_bound_beyond_y>A}. To obtain \eqref{eq:triangle_inequality_bound_exponent}, set $t=\frac{\theta}{\pi}\in[-1,1)$ and reindex the tail terms by odd integers. For \eqref{eq:bounding_gaussian_sum}, write $u=\frac{1-t}{2}$ so that the bracket becomes $1+\sum_{m\in\mathbb Z}\exp(-2A^2(m+u)^2)$. By Poisson summation,
\[
\sum_{m\in\mathbb Z}\exp(-2A^2(m+u)^2)
=
\frac{\sqrt{\pi}}{\sqrt{2}A}
\left(1+2\sum_{n=1}^{\infty}\rme^{-\frac{\pi^2n^2}{2A^2}}\cos(2\pi n u)\right),
\]
which is maximized at $u=0$; equivalently, the bracket is largest at $\theta=\pm\pi$ and equals $2+2\sum_{m=1}^{\infty} e^{-2A^2m^2}$. 

To conclude, we split into two cases. If $A\ge 1$, then $m^2 \ge m$ for $m\ge 1$ gives
\[
\sum_{m=1}^{\infty} \rme^{-2A^2m^2}
\le
\sum_{m=1}^{\infty} \rme^{-2A^2m}
=
\frac{\rme^{-2A^2}}{1-\rme^{-2A^2}},
\]
so
\[
\frac{A}{\pi} \frac{2\rme}{ \sqrt{2 \pi \rme + 4 A^2}}\left(1+\sum_{m=1}^{\infty} \rme^{-2A^2m^2}\right)
\le
\frac{2\rme A}{\pi \sqrt{2 \pi \rme + 4 A^2}}\frac{1}{1-\rme^{-2A^2}}.
\]
Using $1-\rme^{-x} \ge \frac{x}{1+x}$ for $x\ge 0$, we obtain
\[
\frac{1}{1-\rme^{-2A^2}} \le \frac{1+2A^2}{2A^2},
\]
hence
\[
\frac{2\rme A}{\pi \sqrt{2 \pi \rme + 4 A^2}}\frac{1}{1-\rme^{-2A^2}}
\le
\frac{\rme}{\pi}\frac{1+2A^2}{A\sqrt{2\pi \rme+4A^2}}
\le
\frac{\rme}{\pi},
\]
where the last step follows from $(1+2A^2)^2 \le A^2(2\pi\rme+4A^2)$ for $A\ge 1$.

If $0<A\le 1$, then the function $x\mapsto \rme^{-2A^2x^2}$ is decreasing on $[0,\infty)$, so
\[
\sum_{m=1}^{\infty} \rme^{-2A^2m^2}
\le
\int_{0}^{\infty}\rme^{-2A^2x^2}\,\rmd x
=
\frac{\sqrt{\pi}}{2\sqrt{2}A}.
\]
Therefore
\[
\frac{A}{\pi} \frac{2\rme}{ \sqrt{2 \pi \rme + 4 A^2}}\left(1+\sum_{m=1}^{\infty} \rme^{-2A^2m^2}\right)
\le
\frac{2\rme A+\rme\sqrt{\pi/2}}{\pi\sqrt{2\pi\rme+4A^2}}.
\]
Finally, since $A\le 1$,
\[
\left(2A+\sqrt{\pi/2}\right)^2
=
4A^2+4A\sqrt{\pi/2}+\pi/2
\le
4A^2+4\sqrt{\pi/2}+\pi/2
<
4A^2+2\pi\rme,
\]
which implies
\[
\frac{2\rme A+\rme\sqrt{\pi/2}}{\pi\sqrt{2\pi\rme+4A^2}}
\le
\frac{\rme}{\pi}.
\]

This concludes the proof of (P$_4$) and of the proposition.}

\end{proof}

\bibliography{refs.bib}
\bibliographystyle{IEEEtran}

\end{document}